\documentclass[11pt,oneside]{article}
\usepackage{authblk}
\usepackage{amsthm, amsmath, amssymb, bbm, hyperref}
\usepackage{graphicx, wrapfig, caption, subcaption}
\graphicspath{ {./img/} }
\allowdisplaybreaks

\usepackage[dvipsnames]{xcolor}

\def\E{\mathbb{E}}

\newtheorem{theorem}{Theorem}[section]
\newtheorem{definition}{Definition}[section]
\newtheorem{proposition}{Proposition}[section]
\newtheorem{remark}{Remark}[section]
\newtheorem{lemma}{Lemma}[section]

\title{Stability in stochastic hypergraph matching I: necessary and sufficient criteria}
\author{Doan Dai Nguyen, Ana Bu\v{s}i\'{c}}
\affil{Inria and DI ENS, École Normale Supérieure\\PSL University, Paris, France}

\begin{document}
	\maketitle
	\begin{abstract}
		Stochastic matching on hypergraphs is an important topic for its versatility in capturing real-life systems, from living donor transplant to ride-hailing. Nevertheless, finding necessary and sufficient criteria for stability is a long-standing problem. One of the key difficulties is the fact that greedy policies, whilst maximally stable for stochastic matching on graphs, no longer achieve maximal stability region on hypergraphs. So far, no alternative families of policies with similar properties have been known.
		
		In this work, we introduce online assignment policies, in which each item is assigned to a matching hyperedge type upon arrival. We show that this is a good generalisation to greedy policies, by proving that they are maximally stable. Their natural amenability to analysis allow us to derive several necessary and sufficient criteria for stability, which generalise the known criteria for graphs. Furthermore, the constructive proof gives a maximally stable arrival-rate agnostic policy.
	\end{abstract}
	
	\section{Introduction}

\subsection{Motivation}

This article is the first in a series concerning the stability of dynamic matching problem on a hypergraph $G$, where each vertex $v$ denotes a class of items, arriving randomly and waiting in the buffer. A collection of items can form a type-$e$ matching if their corresponding classes form a hyperedge $e$. If multiple matchings are possible, one may need to specify a policy in order to decide which matchings to make; the goal of the matching policy is to not let any items wait in the buffer indefinitely.

This model is motivated by the living-donor transplant problem. Ideally, a patient in need of an organ receives it from someone with close connection, such as blood relatives, a spouse, a friend, or a co-worker. But the donor and the intended recipient might have different characteristics leading to biological incompatibility, such as blood type or antibodies.

In this case, a solution is multi-party organ exchange \cite{Roth2007}, where each exchange takes place amongst a group of $n$ donor-recipient pairs, and for $i = 1, 2, \ldots, n$, the $i$\textsuperscript{th} donor gives their organs to the $(i+1)$\textsuperscript{th} recipient (as usual, we consider by convention the $(n+1)$\textsuperscript{th} pair to be the first pair). In this setting, each vertex in the compatibility graph $G$ is a class indicating some donor-recipient pairs according to their characteristics, and each exchange is a subset of vertices, meaning a hyperedge, resulting in a hypergraph. Naturally, donor-recipient pairs arrive randomly and we desire a matching policy guaranteeing that no patient is expected to wait indefinitely.

When restricting to two-party exchange, $G$ is a graph and we have necessary and sufficient condition for the existence of such policies, as well as their construction \cite{Busic2013, Mairesse2016}. However, Roth, Sönmez, and Ünver demonstrated that allowing three-party exchange decreases waiting time, and such improvement is intrinsic, in the sense that it cannot be achieved by any two-party matching policies \cite{Roth2007}. 

Moreover, in many applications, multi-party exchange is an intrinsic property of the system and cannot be removed. In an assemble-to-order system, a product may require three or more parts \cite{Gurvich2014}. In case of online games, three or more players may be required to form a team \cite{Veron2014}. In case of ride-sharing, multiple drives may be needed for one long journey \cite{Teubner2015}.

As much as it is needed and desired, allowing three or more parties to participate in an exchange poses a difficult problem regarding the design of policy and/or the compatibility hypergraph. In particular, there has been no easy way to check if a given matching model is stabilisable, and if so, how to construct a stabilising policy.

\subsection{Difficulties}

The path towards such a criterion is well-known. In principle, one starts with a class of polices which are, at the same time, simple enough to be analysed and powerful enough to stabilise the model whenever it is stabilisable. Then, simplifying the condition for the existence of such policy within said class eventually leads to a stability criterion which depends only on the matching model itself.

In case of graphs, one may consider greedy matching policies which perform matchings whenever they are available (note that a priori, we are not obliged to do so). Bu\v{s}i\'{c}, Gupta, and Mairesse first demonstrated that this class, and in particular Match-the-Longest policy, has desired properties in case of stochastic matching on bipartite graphs \cite{Busic2013}, which was later generalised to general graphs \cite{Mairesse2016}.

This leads to the first stability criterion for bipartite graph case \cite[Lemma 3.2, Theorem 7.1]{Busic2013} and later for general graph case \cite[Theorem 1]{Mairesse2016}, which both depend on the topology of compatibility graph. Based on this, Comte et al. later proved other equivalent criteria based on the incidence matrix \cite[Proposition 7 and 8]{Comte2021}.

However, in case of hypergraphs, finding a class of polices with the two desired properties above turns out to be difficult: in general, greedy policies are no longer sufficient \cite[Figure 1(right)]{Kerimov2021}. So far, no such class of policies has been identified.

\subsection{Contributions and organisation}

Our contributions in this article are as follow:

\begin{itemize}
	\item In Section \ref{Criterion for online assignment policies}, we introduce a class of randomised policies, called online assignment policies, which fix the matching type for each item upon arrival.
	
	Note that this is different from greedy policies: it is possible that there are some possible matchings if one considers the buffer's content alone and ignores the assigned matching types, but those matchings are not realised due to the assigned types.
	
	\item We then develop a framework to study this policy class which leads to Condition \eqref{eq:onlcond} necessary and sufficient criterion for the existence of a stabilising policy \emph{within this class} (Theorem \ref{theorem: necessity and sufficiency of onlcond}).
	
	\item In Section \ref{Criterion for general policies}, we extend this framework to incorporate all policies and derive Condition \eqref{eq:gencond} necessary and sufficient criterion for stability in general (Theorem \ref{theorem: necessity and sufficiency of gencond}).
	
	\item Last but not least, in Section \ref{Equivalence}, we introduce Condition \eqref{eq:inccond} analogous to that of Comte et al. \cite[Proposition 7 and 8]{Comte2021}, and we show that in fact, all three criteria are equivalent (Theorem \ref{theorem: necessity and sufficiency of inccond}). This implies that our policy class has maximum stability region, as desired.
	
	\item In Section \ref{Graphs vs hypergraphs}, to further demonstrate that our policy class is an appropriate generalisation of greedy polices, we show directly the equivalence between \texttt{NCOND} and Condition \eqref{eq:onlcond} in case of graphs.
	
	\item Finally, in developing the framework, we introduce a function $L_+$ and show that $L_+$-MaxWeight, which is deterministic, also has maximum stability region. 
\end{itemize}

Aside from the aforementioned, Section \ref{Preliminaries} introduces notation and setting of stochastic matching model on hypergraphs, and Section \ref{Conclusion} concludes with some final remarks.

A preliminary version of this work was presented in 24\textsuperscript{th} Conference of the International Federation of Operational Research Societies, 12-17 July 2026, Vienna, Austria \cite{Busic2026} and LINCS Annual Workshop, 2-3 July 2026, Palaiseau, France \cite{Nguyen2026}.

\subsection{Literature review}

\subsubsection{Stability in case of hypergraphs}

In stark contrast with the graph case, little is known about hypergraph cases. Harrison considered a special case, where there is only one hyperedge containing all the vertices, and showed that this model is not stabilisable \cite{Harrison1973}.

The problem of stochastic matching on hypergraphs was first formally introduced by Rahme and Moyal \cite{Rahme2021}, where they considered greedy policies and introduced many different criteria which are natural generalisations of that of Mairesse and Moyal for graphs \cite[Theorem 1]{Mairesse2016}. However, no criteria were shown to be necessary and sufficient, and the question of finding a necessary and sufficient condition for stability was left open.

Diffusion limit approximation has been established by Xie in his PhD thesis \cite[Section 3 and 5]{Xie2022} in some heavy traffic regime, both when items are non-reneging \cite[Theorem 4]{Xie2022} and reneging \cite[Theorem 11]{Xie2022}. Some properties of the limiting process were studied, such as moment bound on the waiting time (under some appropriate scaling) \cite[Thereom 8]{Xie2022}. However, unlike the classical queueing system where there is a link between the stability of original queue model and that of fluid limit, it is not clear that the same holds for matching model.

An equivalent variant of Condition \eqref{eq:inccond} was obtained independently and concurrently by Mathieu \cite{Mathieu2026}, where the solution $\lambda$ to the conservation equation $G\lambda = \mu$ is shown to be non-negative and the incidence matrix $G$ must admit a surjective sub-matrix induced by $\lambda$. In contrast, we show that it suffices for $G$ to be surjective independently from $\lambda$, and there exists $\lambda$ strictly positive. The proof of sufficiency in the former case was done using the Virtual-Queue Match-the-Longest policy \cite{Mathieu2026}, a rewardless variant of the Extended Greedy Primal–Dual policy of Nazari and Stolyar \cite{Nazari2019}.

\subsubsection{Stochastic matching beyond stability}

There have been attempts to study stochastic matching in hypergraphs whilst bypassing the question of stability altogether. However, they all assume explicitly stability or some stronger assumptions. It thus appears that the question of stability is a true bottleneck: without resolving it, further studies are limited.

Notably, Nazari and Stolyar studied the problem of maximising reward in discounted infinite-horizon setting with discounted rate being small enough \cite{Nazari2019}. In particular, they exhibited a policy and showed that it is asymptotically optimal. However, the stability of their policy depends on the assumption that the system is stabilisable \cite[Assumption 5]{Nazari2019}.

Also concerning reward maximisation but under general position condition, Kerimov et al. \cite{Kerimov2021} and Gupta \cite{Gupta2024} independently presented policies achieving asymptotic hindsight optimality with constant regret, which is stronger than asymptotic average-cost optimality. However, general position condition is restrictive and not always satisfied \cite[Section 6.4.1]{Comte2021}.

In another direction, Gurvich and Ward considered the problem of holding-cost minimisation in finite horizon setting, but the optimality supposes a condition on the compatibility graph \cite[Assumption 1]{Gurvich2014}, which, as we will see, is sufficient but not necessary for stability.

\subsubsection{Online assignment policies}

It is worth noting that the notion of online assignment policies that we consider in this article has already been considered in different contexts.

The idea of prescribing matching types was already present in the Virtual Matching System of Nazari and Stolyar \cite[Section 3.3]{Nazari2019}. This is different from our notion of online assignment in that we assign the matching type to each individual item at the moment of arrival, whereas Nazari and Stolyar only draws matching types to be fulfilled by some items in the future, but the matching type for each item is not fixed at the arrival.

Nevertheless, their assumption of stability (Assumption 5, which is a special case of Assumption 8) is somewhat similar to - if not already hints at - our first stability criterion. However, the question of stability was not addressed in the work of Nazari and Stolyar.

Independent from and most similar to our work is that by Gupta \cite{Gupta2024}, where he also considered assigning matching type to each item upon arrival. Moreover, his construction of Lyapunov function is also similar to ours, and his policy is also induced by minimising this function at each step. However, there are three decisive differences.
\begin{itemize}
	\item In work of Gupta, analysis was carried under General Position Condition, which, as we remarked above, is not always satisfied; our work does not assume this condition. Moreover, the goal of the two works are also different: the former aimed toward asymptotically optimal policy in maximising reward. The question of stability was not treated, and no stability criteria were introduced.
	
	\item For each hyperedge $e$, Gupta considered an arbitrary but fixed permutation $(v_1, \ldots, v_{|e|})$ of $e$. Then, denote by $q_{v, e}$ the number of class-$v$ items assigned to hyperedge $e$, Gupta considered the cyclic sum-of-square term
	\[
	L_{1, e} (q) = \sum_{i = 1}^{|e|} (q_{v_{i+1}, e} - q_{v_i, e})^2,
	\]
	where $v_{|e| + 1} = v_1$, then summing $L_{1, e}(q)$ for all $e$. In our case, we consider sum-of-square of all pairwise differences
	\[
	L_{2, e} (q) = \sum_{\{u, v\} \in e^2} (q_{u, e} - q_{v, e})^2
	\]
	where the sum is over the \textit{unordered} pairs of vertices $(u, v)$ in a hyperedge $e$, then summing $L_{2, e}(q)$ for all $e$. Although seemingly small, this allows capturing more precisely the imbalance within a hyperedge, which is a decisive factor leading to the establishment of stability criteria.
	
	\item Another novelty in our work is the definition of assignment rate (Definition \ref{definition: assignment rate}) and reassignment rate (Definition \ref{definition: reassignment rate}). These two notions are what allows us to carry out the analysis for online assignment policies and general policies rigorously.
\end{itemize}
	\section{Preliminaries}
\label{Preliminaries}

\subsection{Matching models}

To define a matching model $(G, \mu)$, we have a hypergraph $G = (V, E)$, so-called a compatibility hypergraph or a matching hypergraph. The vertices $v$ denote classes, and the hyperedges $e$ denote matching types. In this setting, a hyperedge $e \subseteq V$ is a non-empty subset of vertices, or in other words, $E \subset \mathcal{P}(V)$ is a subset of the power set of $V$. In particular, if we have $|e| = 2$ for all $e \in E$, then $G$ is a graph, and moreover, in this setting, since we only concern with the subsets but not the order of elements within a subset, if $G$ is a graph, then $G$ is undirected.

By abuse of notation, we see $G$ also as a subset of $V \times E$, and we say that $(v, e) \in G$ if $e \in E$ and $v \in e$.

We equip it with a buffer $\hat{Z} = (\hat{Z}_v)_{v \in V}$, and a measure $\mu$ of support $V$, meaning $\mu_v > 0$ for all $v \in V$.

A realisation of the model is as follow: we start with some initial buffer state. Then, class-$v$ items arrive according to an independent Poisson process of rate $\mu_v$ and wait in the buffer to be matched. Since scaling $\mu_v$ for all $v$ by an absolute constant is equivalent to scaling time, by convention, we normalise $\mu_v$ by the sum $\sum_{u \in V} \mu_u$, meaning that $\mu$ is now a distribution on $V$, called the arrival rates.

As the total arrival time is state-independent, using classical uniformisation argument, one may assume that at each time $t \in \mathbb{N}$, exactly one item arrive. Intuitively, $\mu_v$ denotes the probability that an item of a given class arrives at a given instance. In other words, at time $t \in \mathbb{N}$, a vertex $v \in V$ is drawn i.i.d. according to $\mu$, and an item of class $v$ arrives and waits in the buffer.

Then, if for some hyperedge $e$ and for each vertex $v \in e$, there is at least one item of class $v$ in the buffer, we say that a matching of type $e$ can be made. However, this is \emph{not} an obligation: we may nevertheless choose not to realise this match, but if we do, for each $v \in e$, exactly one item of class $v$ will depart from the buffer. If we choose to always match whenever there are available matchings, we say that we are matching greedily.

However, it can happen that there are multiple available matchings, and even when we choose to realise a given matching $e$, for a given $v \in e$, there can be multiple items of class $v$. A matching policy therefore must be specified to determine which matchings to make and which items to be matched.

\subsection{State space description}
\label{Matching policies}

To formalise the description in the previous subsection, first let $(A_t)_{t \in \mathbb{N}}$ denote the arrival process: in particular, $A_t$ is a random variable taking value in $V$ and being drawn i.i.d. according to $\mu$. 

Next, each $\hat{Z}_v(t)$ is seen as a (finite) \textit{multiset} of integers denoting the time of arrival of each class-$v$ item; the need for multiset stems from the definition of initial buffer state, where we might already have multiple items of the same class in the buffer. The buffer state $\hat{Z}(t)$ at time $t$ is then seen as an element of $\left(\mathbb{N}^\mathbb{N}\right)^V = \mathbb{N}^{\mathbb{N} \times V}$. 

\begin{definition}
	A type-$e$ matching $M = (M_v)_{v \in V}$ is defined as an element in $\left(\mathbb{N}^\mathbb{N}\right)^V$ such that $|M_v| = 1$ if and only if $v \in e$, and $|M_v| = 0$ otherwise.
	
	Moreover, let $B \in (\mathbb{N}^\mathbb{N})^V$ be a buffer state, then a matching $M$ is admissible if for all $v \in V$, we have $M_v \subset B_v$. If we choose to realise the matching $M$, then the resulting buffer state is $B \setminus M = (B_v \setminus M_v)_{v \in V}$.
\end{definition}
In rough terms, this says that a matching $M$ is a choice of items, whereby we choose only and exactly one class-$v$ item for all $v \in e$, whose arrival time is denoted by $M_v$. Moreover, a matching is admissible if it can be realised.

A matching policy $\Phi$ is then a collection of matching functions $(\Phi_t)_{t \in \mathbb{N}}$, where $\Phi_t$ determines the matchings at time $t$. Then, let $B(t) = (B_v(t))_{v \in V}$ be the buffer state before the matching policy making decision, 
\[
	B_v(t) = 
	\begin{cases}
		\hat{Z}_v(t-1) \cup \{t\} & \text{ if } v = A_t \\
		\hat{Z}_v(t-1) & \text{ otherwise}
	\end{cases},
\]
then $\Phi_t$ is now formally defined as a function taking $\mathcal{H}_t$ as the argument and returning a state space $\hat{Z}_t$, which must be admissible, meaning it must be obtainable from $B(t)$ by realising some (possibly zero) admissible matchings.

Naturally, we require that $\Phi$ only has access to the current history $\mathcal{H}$.

\begin{definition}
	At time $t$, the collection $\mathcal{H}_t$ of random variables $(A_i)_{i \leq t}$ and of functions $(\Phi_i)_{i \leq t-1}$ is called the history up to time $t$.
\end{definition}

In rough terms, it represents everything we know up to time $t$: the arrivals and the decisions made by the policy, which includes the matching made but possibly more; in case of online assignment polices (to be defined in Subsection \ref{Introduction of online assignment policies}), this also includes the matching type assignments. We then require that $\Phi_t$ depends only on $\mathcal{H}_t$.

Note that we choose to include $(\Phi_i)_{i \leq t-1}$ and not $(\hat{Z}(i))_{i \leq t-1}$ because on the one hand, this allows $(\Phi_i)_{i \leq t-1}$ to be random variables modelling a randomised policy, and on the other hand, if $(\Phi_i)_{i \leq t-1}$ is deterministic, $(\hat{Z}(i))_{i \leq t-1}$ will be uniquely determined, so we are at no loss of generality.

Finally, $\Phi$ is said to stabilise $(G, \mu)$ if along a given trajectory, at any point $T \geq 0$, we have that the expected time for the system to clear, that is
\[
	\tau_T = \min \left\{ t \geq 0 \mid \forall v \in V, \hat{Z}_v(t + T) = \emptyset \right\}
\]
is finite, meaning for all $T \geq 0$, we have $\E[\tau_T \mid \mathcal{H}_T] < \infty$; of course, the upper bound might not be uniform. We say that $(G, \mu)$ is stabilisable if there exists such policy $\Phi$.

Indeed, this is a desirable property: this implies that an item once arrived, regardless of the current buffer state, will be served in finite time on average. In case of organ exchange, it means that all patients will receive suitable organs in finite time. In case of assembly-to-order system, this means no parts will be stored indefinitely.

\subsection{Markovian policies, stationary policies, size-based policies}

The previous subsection describes rigorously the notion of matching policies in generality. In practice, we may allow some relaxations.

\begin{definition}
	\hfill
	\begin{itemize}
		\item Given the buffer state $\hat{Z}(t)$ at time $t$, let $Z(t) = (Z_v(t))_{v \in V}$ be given by
		\[
		Z_v(t) = |\hat{Z}_v(t)|,
		\]
		then a matching policy $\Phi$ is said to be size-based if $\Phi_t$ depends only on $(Z(i))_{i \leq t-1}$ and not $(\hat{Z}(i))_{i \leq t-1}$.
		\item A matching policy $\Phi$ is said to be Markovian if $\Phi_t$ depends only on $t$, $\hat{Z}_{t-1}$ and $A_t$. Moreover, such policy is said to be stationary if $\Phi_t$ does not depend on $t$.
	\end{itemize}
\end{definition}

In rough terms, a size-based matching policy does not take into account the order of arrival and treat all class-$v$ items as equal. Many classical policies fall into this category, such as Match-the-Longest or MaxWeight. However, there are also examples of matching policies where the order of arrivals of items is taken into account, such as First-Come-First-Match, where an item of class $v$ is matched as soon as there is an available matching of type $e$ using class $v$.

If $\Phi$ is Markovian and size-based, then $Z$ is a Markov chain, and the fact that $\Phi$ stabilises $(G, \mu)$ is equivalent to saying that $Z$ is positive recurrent.

Stability is a central notion, for aside from guaranteeing finite expected waiting time, which is desirable for real-life applications such as organ exchange or ride-sharing, it also allows rigorous theoretical analysis. In particular, under stability, we can define the matching rate $\lambda_e \geq 0$ along a hyperedge $e \in E$, as the long-run average number of matches of type $e$ per unit of time.

Denote by $A = (A_{v, e})_{v \in V, e \in E}$ the incidence matrix of $G$ where $A_{v, e} = 1$ if and only if $v \in e$, and $A_{v, e} = 0$ otherwise. Per unit of time, there are $\mu_v$ items of class $v$ arriving, and under the assumption of stability, they will all be matched. As the number of matched items of class $v$ is precisely $\sum_{e \ni v} \lambda_e$ by definition, we have $\mu_v = \sum_{e \ni v} \lambda_e$, or in other words, $A \lambda = \mu$. This is called the conservation equation.
	\section{Online assignment policies and first stability criterion}
\label{Criterion for online assignment policies}

\subsection{Online assignment policies}
\label{Introduction of online assignment policies}

To define online assignment policies, we first define the notion of assignment.

\begin{definition}
	If a class-$v$ item is said to be assigned to the hyperedge $e \ni v$, it means that this item then can only be matched by a matching of type $e$. In particular, even if there are other available matchings of some type $f$ with $v \in f$, this item will still not be matched.
\end{definition}

Then, we enlarge the buffer to have one dedicated buffer $\hat{X}_{v, e}$ for each pair $(v, e)$ with $e \in E$ and $v \in V$. Formally, as in Section \ref{Matching policies}, $\hat{X}_{v, e}(t)$ is an element in $2^\mathbb{N}$, and $\hat{X}(t)$ is thus an element in $(2^\mathbb{N})^{V \times E}$.

Next, we define the notion of assignment rate.

\begin{definition}
	Given an arrival rate $\mu$, w call a tuple $\nu = (\nu_{v, e})_{(v, e) \in G}$ an assignment rate if it satisfies the following two conditions:
	\begin{itemize}
		\item For all $(v, e) \in G$, $\nu_{v, e} \geq 0$;
		\item For all $v \in V$, $\sum_{e \ni v} \nu_{v, e} = \mu_v$.
	\end{itemize}
\end{definition}

Then, an online assignment policy is defined as follow: at time $t$, we first choose an assignment rate $\nu = \nu(t)$. Note that as in the definition of matching policies in Section \ref{Matching policies}, \emph{a priori} the assignment rate $\nu(t)$ at time $t$ can depend on the whole history of arrivals and matchings up to time $t$.

Then, if a class-$v$ item arrives, we assign it to an edge $e \ni v$ with probability $\frac{\nu_{v, e}}{\mu_v}$. This explains the conditions for $\nu$: they are to ensure that for each $v \in V$, we have that $\left(\frac{\nu_{v, e}}{\mu_v}\right)_{e \ni v}$ is a probability distribution on the set of hyperedges $e$ containing $v$, so that the policy is well-defined. This also explains the choice of the name ``assignment rate'': with probability $\nu_{v, e}$, a class-$v$ item arrives and can only be matched by a type-$e$ matching.

An online assignment policy $\Phi$ is said to stabilise $(G, \mu)$ if at any given moment $T$, the expected time to clear the system is finite. In other words, for $T \geq 0$, let
\[
	\tau_T = \min \left\{ t \geq 0 | \forall (v, e) \in G, \hat{X}_{v, e}(t + T) = \emptyset\right\},
\]
then $\E[\tau_T \mid \mathcal{H}_T] < \infty$, in which case we say that $(G, \mu)$ is stabilisable within the class of online assignment policy.

If $\Phi$ is Markovian and size-based, meaning that $\nu(t)$ depends only on $t$, $A_t$ and $X(t)$, then $X$ is a Markov chain, and the stability of $(G, \mu)$ is equivalent to saying that $X$ is positive recurrent. Note that by definition, $\sum_{e \ni v} X_{v, e}(t) = Z_v(t)$, so if $X$ is positive recurrent, so is $Z$.

\subsection{Necessary condition for stability}
\label{Necessity for stability for online assigment policies}

Before getting to the actual condition, we first revisit the matching problem in case of graph, where the necessity (and sufficient) condition, so-called \texttt{NCOND}, was given by Mairesse and Moyal \cite{Mairesse2016}. In particular, when $G$ is a graph, $(G, \mu)$ is stabilisable if and only if for all independent set $U$, we have
\[
	\mu(U) < \mu(N(U)),
\]
where $N(U) = \{v \mid (u, v) \in E\}$ is the neighbourhood of $U$. In rough terms, per unit of time, there are $\mu(U)$ items arriving in $U$, and $\mu(N(U))$ items arriving in $N(U)$. If $U$ is the support of the buffer at time $t$, meaning $U = \{v \mid X_v(t) > 0\}$, then each time arriving in $N(U)$ can (and will) be matched with an item in $U$, so the total number of item in expectation increases by $\mu(U) - \mu(N(U))$. What \texttt{NCOND} says is that this quantity is negative for all possible support of the buffer, so regardless the current state, in expectation, the number of items decreases, so the system is stabilisable.

Back to the hypergraph case, we can find heuristically a similar condition. Indeed, consider the buffer state $X(t)$ at time $t$. Within one hyperedges $e$, for a given vertex $v \in e$, there are $X_{v, e}(t)$ class-$v$ items assigned to $e$, which can only leave by type-$e$ matchings. From the items in the buffer, there are precisely $\min_{v \in e} X_{v, e}(t)$ type-$e$ matchings available, and suppose we have heuristically made all of them, then we can assume
\[
	\min_{v \in e} X_{v, e}(t) = 0.
\]
Then, we can consider the set $U = \{(v, e) \mid X_{v, e}(t) > 0\} \subset G$ the support of the buffer. 

\begin{definition}
	\label{definition: assignment rate}
	A set $U \subset G$ is called a support if for each $e \in E$, the $e$-slice of $U$, given by $U_e = \{v \mid (v, e) \in U\}$, cannot be $e$, or in other words, $U_e \subsetneq e$.
\end{definition}

By the condition $\min_{v \in e} X_{v, e}(t) = 0$, the support of a buffer is a support in the sense of the previous definition. Conversely, for such a set $U$ satisfying the previous definition, we can construct a state $X(t)$ such that $U$ is the support of $X(t)$, simply by letting $X_{v, e}(t) = 1$ if $(v, e) \in U$ and $0$ otherwise. This plays the role of ``independent set'' in the graph analogy.

Next, suppose we choose an assignment rate $\nu = \nu(t)$ at time $t$, for a given pair $(v, e)$, in expectation, there are $\nu_{v, e}$ class-$v$ items arrived and assigned to $e$. On the other hand, a type-$e$ matching requires $|e|$ items; for each $u \in e$, we can think of each arrival of class-$u$ items as bringing $\frac{1}{|e|}$ matching. Each matching reduces the number of class-$v$ items assigned to $e$ by 1, so in expectation, the change of this number is roughly  
\[
	\nu_{v, e} - \frac{\nu_e}{|e|},
\]
where $\nu_e = \sum_{e \ni u} \nu_{u, e}$.

Similar to the \texttt{NCOND} above, we can expect an analogous condition.
\begin{equation}
	\label{eq:onlcond}
	\textrm{\parbox{.8\textwidth}{For all supports $U$, there exists an assignment rate $\nu^U$ such that for all $(v, e) \in U$,
		\[
		\nu^U_{v, e} < \frac{\nu^U_e}{|e|}.
		\]
	}}
\end{equation}

This is of course only an informal argument and some points are not satisfactory: Is it true that we can consider all class-$u$ items as equal? What does $\frac{1}{|e|}$ of a type-$e$ matching mean? Why should we have the number of class-$v$ items assigned to $e$ decrease for \emph{all} $(v, e) \in U$ in the support, whereas in \eqref{eq:onlcond}, we only have the \emph{total} number of items decrease in expectation?

Fortunately, the simple behaviour of online assignment policies allows amenable analysis. Indeed, \eqref{eq:onlcond} is a necessary condition.

\begin{lemma}
	\label{lemma: necessity, online assignment}
	The condition \eqref{eq:onlcond} is necessary for stability within the class of online assignment policies: if $(G, \mu)$ is stabilisable by an online assignment policy, then \eqref{eq:onlcond} is satisfied. 
\end{lemma}
\begin{proof}
	To make precise our informal argument above, we fix an online assignment policy $\Phi$ stabilising $(G, \mu)$ and define an auxiliary sequence of random variable $Y = (Y_{v, e}(t))_{e \in E, v \in e, t \in \mathbb{N}}$. Firstly, we assign the matching types for the items in the initial buffer in an arbitrary manner, prescribed by the current online assignment policy. Then, suppose at time $t$, a class-$v$ item arrives, and our current online assignment policy chooses an assignment rate $\nu = \nu(t)$, which then chooses this item to the hyperedge $e$. We have
	\[
	Y_{u, f}(t+1) = 
	\begin{cases}
		Y_{u, f}(t) + 1 - \frac{1}{|f|} & \text{ if } u = v \text{ and } f = e\\
		Y_{u, f}(t) - \frac{1}{|f|} & \text{ if } u \neq v \text{ and } f = e\\
		Y_{u, f}(t) & \text{ otherwise}
	\end{cases}.
	\]
	
	To start, we show that $Y$ is a lower bound of $X$ in some appropriate sense. Indeed, by induction, we can show that for all $f \in E$ and $u \in f$, we have
	\[
		Y_{u, f}(t) = X_{u, f}(t) - \frac{1}{|f|} \sum_{f \ni v} X_{v, f}(t),
	\]
	so in particular, if $X(t) = 0$ for some $t$, we also have $Y(t) = 0$. Thus, since $(G, \mu)$ is stabilised, we have $Y(\tau) = 0$ for some $\tau$ positive and of finite mean.
	
	Now, assume for the sake of contradiction that \eqref{eq:onlcond} is not satisfied, meaning that there exists a support set $U \subset \{(v, e) \mid e \in E, v \in e\}$ such that for all $e \in E$, $U_e = \{v \mid (v, e) \in U\} \subsetneq e$, and such that for all assignment rate $\nu$, we have
	\[
		\nu_{v, e} \geq \frac{\nu_e}{|e|},
	\]
	for some $(v, e) \in U$.
	
	Given an assignment rate $\nu$, for $(v, e) \in U$, we define
	\[
		(\Delta \nu)_{v, e} = \nu_{v, e} - \frac{\nu_e}{|e|},
	\]
	which is a linear map acting on $\nu$. The set of assignment rates $\nu$ is a convex and compact polytope $\mathcal{P}$ in $\mathbb{R}^G$, defined by the following constraints:
	\begin{itemize}
		\item For all $(v, e) \in G$, $\nu_{v, e} \geq 0$;
		\item For all $v \in V$, $\sum_{e \ni v} \nu_{v, e} = \mu_v$.
	\end{itemize}
	Since $\Delta$ is a linear map, the image $\Delta(\mathcal{P})$ of $\mathcal{P}$ is also a (convex and compact) polytope, this time living in $\mathbb{R}^U$. The failure of \eqref{eq:onlcond} is equivalent to saying that every point in $\Delta(\mathcal{P})$ has at least one positive coordinate.
	
	In other words, $\Delta(\mathcal{P})$ and $\mathbb{R}_{\leq 0}^U$ have disjoint interior. As both sets are convex, hyperplane separation theorem implies that there exists an affine hyperplane $H$ separating $\Delta(\mathcal{P})$ and $\mathbb{R}_{< 0}^U$. Up to a translation, $H$ can be made to be tangent with $\mathbb{R}_{\leq 0}^U$, but since it is a cone, this means $H$ passes by the origin, or in other words, $H$ is a hyperplane in $\mathbb{R}^U$.
	
	Write $H = \{x \mid \langle w, x \rangle = 0\}$ for a normal vector $w$ of $H$ which is non-zero since $H$ is non-degenerate, the fact that $H$ separates $\Delta(\mathcal{P})$ and $\mathbb{R}_{\leq 0}^U$ means we can choose $w$ such that
	\begin{itemize}
		\item for all $x \in \Delta(\mathcal{P})$, we have $\langle x, w \rangle \geq 0$, and
		\item for all $y \in \mathbb{R}_{\leq 0}^U$, we have $\langle y, w \rangle \leq 0$.
	\end{itemize}
	By replacing $y$ with $-y$, the second property is equivalent to saying that for all $y \in \mathbb{R}_{\geq 0}^U$, we have $\langle y, w \rangle \geq 0$, or in other words, $w$ lies in the dual cone of the positive quadrant $\mathbb{R}_{\geq 0}^U$. But since this is a self-dual cone, we have $w \in \mathbb{R}_{\geq 0}^U$, or $w$ has only non-negative coordinates.
	
	Now, we define a function $L_-$ such that for a state $c \in \mathbb{R}^G$, we have 
	\[
		L_-(c) = \sum_{(v, e) \in U} w_{v, e} c_{v, e},
	\]
	By construction of $Y$, for $(v, e) \in G$, if $(v, e) \in U$ and a class-$v$ item arrives and is assigned to the hyperedge $e$, we will have
	\begin{align*}
		L_-(Y(t+1)) - L_-(Y(t)) 
		& = w_{v, e}\left(1 - \frac{1}{|e|}\right) - \sum_{(u, e) \in U, u \neq v} \frac{w_{u, e}}{|e|} \\
		& = w_{v, e} - \frac{\sum_{(u, e) \in U} w_{u, e}}{|e|},
	\end{align*}
	and $L_-(Y(t+1)) - L_-(Y(t)) = 0$ if $(v, e) \not\in U$. Therefore, conditioning on $\hat{X}(t)$ and the behaviour of $\Phi$ at time $t$ denoted by an assignment rate $\nu$, we have
	\begin{align*}
		\E\left[L_-(Y(t+1)) - L_-(Y(t))\right] 
		& = \sum_{(v, e) \in U} \nu_{v, e} \left(w_{v, e} - \frac{\sum_{(u, e) \in U} w_{u, e}}{|e|}\right) \\
		& = \sum_{(v, e) \in U} \nu_{v, e} w_{v, e} - \sum_{(v, e) \in U} \nu_{v, e} \frac{\sum_{(u, e) \in U} w_{u, e}}{|e|} \\
		& = \sum_{(v, e) \in U} \nu_{v, e} w_{v, e} - \sum_{(v, e) \in U} w_{v, e} \frac{\sum_{(u, e) \in U} \nu_{u, e}}{|e|} \\
		& = \sum_{(v, e) \in U} w_{v, e} \left(\nu_{v, e} - \frac{\sum_{(u, e) \in U} \nu_{u, e}}{|e|} \right) \\
		& = \langle w, \Delta \nu \rangle \geq 0.
	\end{align*}
	
	In particular, this makes $[L_- (Y(t))]_{t \in \mathbb{N}}$ a submartingale; at time $t = \tau$, we have $Y(\tau) = 0$, so $L_- (Y(\tau)) = 0$. Doob's submartingale inequality then forces $L_-(Y(t)) \leq 0$ for all $t$ almost surely.
	
	However, this is true regardless of the initial buffer state, and in particular $Y(0)$, but we can define $Y(0)$ such that $L_- (Y(0)) > 0$. Indeed, since $w$ is non-zero and has non-negative coefficients, there exists $(u, e) \in U$ such that $w_{u, e} > 0$. Let the initial buffer contain only one class-$u$ items assigned to $e$ and nothing else, we have that
	\[
		Y_{v, f}(0) = 
		\begin{cases}
			1 - \frac{1}{|e|} & \text{ if } v = u \text{ and } f = e \\
			-\frac{1}{|e|} & \text{ if } v \neq u \text{ and } f = e \\
			0 & \text{ otherwise}
		\end{cases}
	\]
	Suppose without loss of generality that $w_{u, e} = \max_{v \in e} w_{v, e}$, we have
	\[
		L_- (Y(0)) = w_{u, e} - \frac{\sum_{(v, e) \in U} w_{v, e}}{|e|} > 0.
	\]
	where the last inequality is because $U$ is a support, so there exists $v \in e$ such that $(v, e) \not\in U$. We thus have our desired contradiction.	
\end{proof}

\begin{remark}
	Once we get the bound
	\[
		\E\left[L_-(Y(t+1)) - L_-(Y(t))\right] = \langle w, \Delta \nu \rangle \geq 0,
	\]
	it is tempting to conclude by Foster's negative criterion, which avoids the last part of the proof. However, this requires assuming that $Y$ is Markov chain, which also means $X$ is a Markov chain, and the proof then is only valid within the class of Markovian online assignment policies.
\end{remark}

\subsection{Condition \eqref{eq:onlcond} is sufficient}
\label{Sufficiency for stability for online assigment policies}

In this subsection, we show that \eqref{eq:onlcond} is also sufficient. The idea is similar to that by Bu\v{s}i\'{c}, Gupta, and Mairesse \cite[Section 7]{Busic2013}: we will introduce a Lyapunov function $L_+$, and a \emph{deterministic} policy of type MaxWeight which tries to minimise $L_+$ at every time step.

Then, given a matching model $(G, \mu)$ which is stabilisable, per the previous subsection, we know that \eqref{eq:onlcond} is satisfied. To show that this policy stabilises $(G, \mu)$, using \eqref{eq:onlcond}, we construct a randomised policy depending on $\mu$ which achieves negative drift with respect to $L_+$ outside a finite region. Since the $L_+$-MaxWeight policy achieves the smallest possible drift, it too also stabilises $(G, \mu)$.

The challenge here lies in the fact that we cannot use the quadratic Lyapunov function $L(c) = \sum_{(v, e) \in U} c_{v, e}^2$ as did Bu\v{s}i\'{c}, Gupta, and Mairesse. Whilst one can see how $L$ changes when an item arrives and is assigned, it is much more difficult to analyse the drift when a matching is realised.

To mitigate this problem, we introduce a new Lyapunov function. For a given state $c \in \mathbb{N}^G$, we define
\[
	L_+(c) = \frac{1}{2} \sum_{e \in E} \sum_{\{u, v\} \in e^2} (c_{u, e} - c_{v, e})^2,
\]
where the second sum is over all \emph{unordered} pair $(u, v)$.

This function captures not the number of items in each category, but rather the imbalance of the buffer within each hyperedge. Or, in other words, it captures how far away we are from the ideal state where all items are matched. Indeed, notice that the term $\sum_{\{u, v\}\in e^2} (c_{u, e} - c_{v, e})^2$ does not depend on the number of matchings realised (or realisable) $\min_{u \in e} c_{u, e}$, so it is agnostic to the number of matching realised (or realisable) and we avoid altogether the problem of analysing the drift when a matching occurs.

\begin{lemma}
	\label{lemma: sufficience, online assignment}
	The condition \eqref{eq:onlcond} is sufficient for stability within the class of online assignment policies: if \eqref{eq:onlcond} is satisfied, then $X$ is positive recurrent. Moreover, if $(G, \mu)$ is stabilisable within this class, then so is it by a stationary size-based online assignment policy.
\end{lemma}
\begin{proof}
	First, assuming that \eqref{eq:onlcond} holds, we introduce a randomised policy. At time $t$, consider the buffer state $X(t)$ and assume that we have realised all the available matchings, so that for all $e \in E$, we have $\min_{u \in e} X_{u, e}(t) = 0$. Let $U = \{(v, e) \mid X_{v, e}(t) > 0\}$ be the support of $X(t)$, by $\eqref{eq:onlcond}$, there exists an assignment rate $\nu$ such that for all $(u, e) \in U$, we have 
	\[
		\nu_{u, e} < \frac{\nu_e}{|e|}.
	\]
	Note that $\nu$ can be chosen to depend only on $U$, and thus on $X(t)$, thus this policy is size-based and stationary.
	
	We then define our randomised policy as the online assignment policy associated with $\nu$: if at time $t$, a class-$v$ item arrives, we assign it to the edge $e \ni v$ with probability $\frac{\nu_{v, e}}{\mu_v}$.
	
	Now for the analysis, consider the buffer state $X(t) = c \in \mathbb{N}^G$ at time $t$. For a given $(u, e) \in G$, assume that a class-$u$ item arrives and is assigned to $e$, we have
	\begin{align*}
		\Delta L_+ 
		& = L_+(X(t+1)) - L_+(X(t)) \\
		& = \frac{1}{2} \left[\sum_{v \in e} 2(c_{u, e} - c_{v, e}) + 1\right] \\
		& = \frac{|e|}{2} + |e| c_{u, e} - c_e,
	\end{align*}
	where $c_e = \sum_{v \in e} c_{v, e}$. For a given assignment rate $\nu$, we have
	\begin{align*}
		\E[\Delta L_+] 
		& = \sum_{(u, e) \in G} \nu_{u, e} \left[\frac{|e|}{2} + |e| c_{u, e} - c_e\right] \\ 
		& = \sum_{(u, e) \in G} \nu_{u, e} \frac{|e|}{2} + \sum_{(u, e) \in G} \nu_{u, e} [|e| c_{u, e} - c_e] \\
	\end{align*}
	The first term can be uniformly bounded by noticing that by definition of assignment rate, $\sum_{(u, e) \in G} \nu_{u, e} = \sum_{u \in V}\mu_v = 1$, which implies
	\[
		\sum_{(u, e) \in G} \nu_{u, e} \frac{|e|}{2} \leq \sum_{(u, e) \in G} \nu_{u, e} \frac{|V|}{2} = \frac{|V|}{2}.
	\]
	For the second term, we rewrite
	\begin{align*}
		\sum_{(u, e) \in G} \nu_{u, e} [|e| c_{u, e} - c_e] 
		& = \sum_{e \in E} \left[\sum_{u \in e} \nu_{u, e} \left(|e| c_{u, e} - c_e\right)\right] \\
		& = \sum_{e \in E} \left[\sum_{u \in e} \nu_{u, e} |e| c_{u, e} - \left(\sum_{u \in e} \nu_{u, e}\right) c_e \right] \\
		& = \sum_{e \in E} \left[\sum_{u \in e} \nu_{u, e} |e| c_{u, e} - \nu_e c_e \right] \\
		& = \sum_{e \in E} \left[\sum_{u \in e} \nu_{u, e} |e| c_{u, e} - \nu_e \sum_{u \in e} c_{u, e} \right] \\
		& = \sum_{e \in E} \left[\sum_{u \in e} c_{u, e} |e| \left(\nu_{u, e} - \frac{\nu_e}{|e|}\right) \right] \\
		& = \sum_{(u, e) \in G} c_{u, e} |e| \left(\nu_{u, e} - \frac{\nu_e}{|e|}\right). \\
	\end{align*}
	Assume now that all matches realisable are realised, meaning that for all $e \in E$, we have $\min_{u \in e} c_{u, e} = 0$. Let $U = \{(v, e) \mid X_{v, e} (t) > 0\}$ be the support of $X$, we can rewrite
	\begin{align*}
		\E[\Delta L_+] 
		& \leq \frac{|V|}{2} + \sum_{(u, e) \in G} c_{u, e} |e| \left(\nu_{u, e} - \frac{\nu_e}{|e|}\right) \\
		& = \frac{|V|}{2} + \sum_{(u, e) \in U} c_{u, e} |e| \left(\nu_{u, e} - \frac{\nu_e}{|e|}\right)
	\end{align*}
	By \eqref{eq:onlcond}, there exists some assignment rate $\nu^U$ such that 
	\[
		\nu^U_{u, e} < \frac{\nu^U_e}{|e|}
	\]
	for all $(u, e) \in U$. Let $\varepsilon_U = \min_{(u, e) \in U} \left(\frac{\nu_e}{|e|} - \nu_{u, e}\right) > 0$, we have
	\[
		\nu^U_{u, e} - \frac{\nu^U_e}{|e|} \leq -\varepsilon_U.
	\]
	Let $\varepsilon_G = \min_U \varepsilon_U > 0$ which is well-defined because $G$ is finite and $U \subset G$, then $\varepsilon_G$ depends only on $G$ and $\mu$, and moreover, 
	\[
		\nu^U_{u, e} - \frac{\nu^U_e}{|e|} \leq -\varepsilon_G,
	\]
	which implies
	\begin{align*}
		\E[\Delta L_+] 
		& \leq \frac{|V|}{2} + \sum_{(u, e) \in U} c_{u, e} |e| \left(\nu^U_{u, e} - \frac{\nu^U_e}{|e|}\right) \\
		& \leq \frac{|V|}{2} - \varepsilon_G \sum_{(u, e) \in U} c_{u, e} |e| \\
		& \leq \frac{|V|}{2} - \varepsilon_G \sum_{(u, e) \in U} c_{u, e} = \frac{|V|}{2} - \varepsilon_G \|c\|_1,
	\end{align*}
	where the last bound is because $e$ is non-empty for all $e \in E$.
	
	Fix $\varepsilon > 0$, for all $c \in \mathbb{N}^G$ such that $\|c\|_1 \geq \frac{1}{\varepsilon_G}\left(\frac{|V|}{2} + \varepsilon\right)$, we have $\E[\Delta L_+] \leq -\varepsilon$, thus by Foster's positive criterion, $X$ is positive recurrent.
\end{proof}

\begin{remark}
	The stationary size-based online assignment policy introduced in the proof of Lemma \ref{lemma: sufficience, online assignment} is randomised. To obtain a deterministic policy, it suffices to remark that since $L_+$-MaxWeight minimises the value of $L_+$, the drift $\E[\Delta L_+] \leq -\varepsilon$ is also achieved with $L_+$-MaxWeight. Therefore, $L_+$-MaxWeight also stabilises $(G, \mu)$.
\end{remark}

Together with Lemma \ref{lemma: necessity, online assignment}, we have
\begin{theorem}
	\label{theorem: necessity and sufficiency of onlcond}
	The condition \eqref{eq:onlcond} is necessary and sufficient for stability within the class of online assignment policies. Moreover, if $(G, \mu)$ is stabilised by an online assignment policy, so is it by $L_+$-MaxWeight.
\end{theorem}
	\section{General policies and second stability criterion.}
\label{Criterion for general policies}

In Section \ref{Criterion for online assignment policies}, we have introduced the class of online assignment policies, as well as demonstrated a necessary and sufficient condition for a matching model $(G, \mu)$ to be stabilisable by a policy in this class. However, it is \textit{a priori} possible that $(G, \mu)$ is not stabilisable by any such policies, but nevertheless is stabilisable by a policy outside the class.

In this Section, we enlarge the framework of online assignment policies to incorporate all possible policies. This modification still allows analysis, and another necessary and sufficient for stability similar to \eqref{eq:onlcond} is devised. As a corollary, we show that $L_+$-MaxWeight has maximum stability region.

\subsection{Reassignment and general policies}

The framework of online assignment policies has already almost captured all the general policies. What is missing is the reassignment, or in other words, the possibility of reassigning items.

\begin{definition}
	\label{definition: reassignment rate}
	Given a support $U$, we call a tuple $\gamma = (\gamma_{v, e, e'})_{v \in e, e'}$ a reassignment rate if it satisfies the following conditions:
	\begin{itemize}
		\item $\gamma_{v, e, e'} \geq 0$;
		\item for all $(v, e) \not\in U$ and $e' \ni v$, we have $\gamma_{v, e, e'} = 0$;
		\item for all $v$, we have $\sum_{e \ni v} \sum_{e' \ni v} \gamma_{v, e, e'} = 1$.
	\end{itemize}
	We say that this reassignment rate $\gamma$ is adapted to $U$.
\end{definition}

Given \emph{any} policy $\Phi$, its behaviour at time $t$ conditioned on $(\hat{X}_{v, e}(i))_{i \leq t-1}$ can be described by a triple $(\alpha, \nu, \gamma)$ where $0 \leq \alpha \leq 1$, $\nu$ is an assignment rate, and $\gamma$ is a reassignment rate adapted to the support $U$ of $X(t)$, as follow:
\begin{itemize}
	\item With probabillity $\alpha$, we accept the arriving item. If a class-$v$ items arrives, we assign it to the hyperedge $e \ni v$ with probability $\frac{\nu_{v, e}}{\mu_v}$;
	\item With probability $1-\alpha$, we temporarily hold on the arrival. Instead, in this case, with probability $\gamma_{v, e, e'}$, we reassign one class-$v$ items from the hyperedge $e$ to $e'$.
\end{itemize}
This last point is why we have the condition that for all $(v, e) \not\in U$ and $e' \ni v$, we have $\gamma_{v, e, e'} = 0$, namely to make sure that we can reassign items only when there are actual items in the buffer to be reassigned.

\begin{remark}
	The special case $\alpha = 0$ corresponds to online assignment policies. Indeed, $1-\alpha$ can be seen as the probability that we choose to reassign an item. 
\end{remark}

\subsection{Necessary condition for stability}

Now we see how to adapt the intuition from Subsection \ref{Introduction of online assignment policies}. As before, if we accept the arriving item, then the change of the number of class-$v$ items assigned to $e$ is 
\[
	\nu_{v, e} - \frac{\nu_e}{|e|}.
\]
If, however, we temporarily hold on the arrival, then the probability that a class-$v$ item is reassigned \emph{to} $e$ is $\sum_{e' \ni v, e' \neq e} \gamma_{v, e', e}$, and similarly, the probability that a class-$v$ item is reassigned \emph{from} $e$ is $\sum_{e' \ni v, e' \neq e} \gamma_{v, e, e'}$, thus the change of the number of class-$v$ items assigned to $e$ in this case is
\[
	\sum_{e' \ni v, e' \neq e} \gamma_{v, e', e} - \gamma_{v, e, e'}.
\]
However, recall that there is also the effect of bringing $\frac{1}{|e|}$ matchings for each item. In particular, this was the reason why we have the term $-\frac{\nu_e}{|e|}$ in case of assignment. 

Now denote $\Gamma_{u, e} = \sum_{e' \ni v, e' \neq e} \gamma_{v, e', e} - \gamma_{v, e, e'}$ and $\Gamma_e = \sum_{u \in e} \Gamma_{u, e}$, a similar reasoning in case of reassignment means we expect the change of the number of class-$v$ items assigned to $e$ in this case to be
\[
	\Gamma_{u, e} - \frac{\Gamma_e}{|e|}.
\]
In total, we have the change of the number of class-$v$ items assigned to $e$ to be
\[
	\alpha \left(\nu_{v, e} - \frac{\nu_e}{|e|}\right) + (1 - \alpha) \left(\Gamma_{u, e} - \frac{\Gamma_e}{|e|}\right).
\]
As in the Subsection \ref{Introduction of online assignment policies}, we also expect that if $(G, \mu)$ is stabilisable by some policy, then this number should be strictly negative for all $(v, e) \in U$ for all support $U$, though the exact choice of $\nu$, $\gamma$, and $\alpha$ may depend on $U$.

Again, as for the online assignment polices, many questions arise: is it justified when multiple reassignments occur between two arrivals? Is it justified to also take into account the effect of bringing $\frac{1}{|e|}$ matchings for each reassigned item?

However, this description does facilitate feasible analysis. As for online assignment policy and condition \eqref{eq:onlcond}, we may conjecture the following condition
\begin{equation}
	\label{eq:gencond}
	\textrm{\parbox{.8\textwidth}{For all supports $U$, there exists an assignment rate $\nu^U$, a reassignment rate $\gamma^U$, and a number $0 \leq \alpha^U \leq 1$ such that for all $(v, e) \in U$,
		\[
			\alpha^U \nu^U_{v, e} + \left(1 - \alpha^U\right) \Gamma_{u, e} < \frac{1}{|e|} \left[\alpha^U \nu^U_e + \left(1 - \alpha^U\right) \Gamma_e\right].
		\]
	}}
\end{equation}

First, we show that this condition is necessary.

\begin{lemma}
	\label{lemma: necessity, general policies}
	The condition \eqref{eq:gencond} is necessary for stability: if $(G, \mu)$ is stabilisable, then \eqref{eq:gencond} is satisfied. 
\end{lemma}
\begin{proof}
	The idea of the proof is very similar to that of Lemma \ref{lemma: necessity, online assignment}, only that we have to adapt the construction of $Y$ to account for reassignment.
	
	we define an auxiliary Markov chain $Y = (Y_{v, e}(t))_{e \in E, v \in e, t \in \mathbb{N}}$. Firstly, $Y_{v, e}(0) = 0$ for all $e \in E$ and $v \in e$. Then, suppose at time $t$, our current online assignment policy chooses an assignment rate $\nu = \nu(t)$, a reassignment rate $\gamma = \gamma(t)$ and a parameter $\alpha = \alpha(t)$ such that $0 \leq \alpha \leq 1$. If a class-$v$ item arrive and our policy then chooses this item to the hyperedge $e$, then we let
	\[
		Y_{u, f}(t+1) = 
		\begin{cases}
			Y_{u, f}(t) + 1 - \frac{1}{|f|} & \text{ if } u = v \text{ and } f = e\\
			Y_{u, f}(t) - \frac{1}{|f|} & \text{ if } u \neq v \text{ and } f = e\\
			Y_{u, f}(t) & \text{ otherwise}
		\end{cases}.
	\]
	If however, it chooses to reassign a class-$v$ item from $e$ to $e'$, then we let
	\[
		Y_{u, f}(t+1) = 
		\begin{cases}
			Y_{u, f}(t) - 1 + \frac{1}{|f|} & \text{ if } u = v \text{ and } f = e\\
			Y_{u, f}(t) + \frac{1}{|f|} & \text{ if } u \neq v \text{ and } f = e\\
			Y_{u, f}(t) + 1 - \frac{1}{|f|} & \text{ if } u = v \text{ and } f = e'\\
			Y_{u, f}(t) - \frac{1}{|f|} & \text{ if } u \neq v \text{ and } f = e'\\
			Y_{u, f}(t) & \text{ otherwise}
		\end{cases}.
	\]
	Now by induction, we can show that for all $f \in E$ and $u \in f$, we have
	\[
		Y_{u, f}(t) = X_{u, f}(t) - \frac{1}{|f|} \sum_{f \ni v} X_{v, f}(t),
	\]
	so it is as if no reassignment ever occurs. Indeed, the modification of $Y$ when a reassignment occurs is precisely to offset the effect of assignment done earlier.
	
	The rest of the proof goes analogously to that of Lemma \ref{lemma: necessity, online assignment}.
\end{proof}

\subsection{Condition \eqref{eq:gencond} is sufficient}

Then, we also show that this condition is sufficient. The proof of this result is analogous to that of Lemma \ref{lemma: sufficience, online assignment} and is omitted for the sake of brevity.

\begin{lemma}
	\label{lemma: sufficience, general policies}
	The condition \eqref{eq:onlcond} is sufficient for stability: if \eqref{eq:gencond} is satisfied, then $X$ is positive recurrent. Moreover, if $(G, \mu)$ is stabilisable within this class, then so is it by a stationary size-based policy and by $L_+$-MaxWeight.
\end{lemma}

Together, we have
\begin{theorem}
	\label{theorem: necessity and sufficiency of gencond}
	The condition \eqref{eq:gencond} is necessary and sufficient for stability. Moreover, $L_+$-MaxWeight has maximum stability region: if $(G, \mu)$ is stabilisable, then it is stabilised by $L_+$-MaxWeight.
\end{theorem}

\begin{remark}
	It should be noted that the two $L_+$-MaxWeight policies in Lemma \ref{lemma: sufficience, online assignment} and Lemma \ref{lemma: sufficience, general policies} are not the same. In particular, the former is within the class of online assignment policies, meaning it is not allowed to reassign items, whereas the latter is. A priori, the latter may seem more powerful than the former in terms of stability region, although we shall see that it is not the case.
\end{remark}

	\section{Equivalence of \eqref{eq:onlcond} and \eqref{eq:gencond}, third stability criterion}
\label{Equivalence}

In Section \ref{Criterion for general policies}, we have shown that \eqref{eq:gencond} is necessary and sufficient for stability, and moreover, $L_+$-MaxWeight has maximum stability region. That still begs two questions:
\begin{itemize}
	\item What is the relationship between \eqref{eq:onlcond} and \eqref{eq:gencond}? As the former is a special case of the latter, the latter is weaker than the former, but is it strictly weaker? Or, in other words, is there an online assignment policy with maximum stability region?
	
	Note that $L_+$-MaxWeight is \emph{not necessarily} online assignment. In particular, it is permitted to reassignment items to minimise $L_+$.
	
	\item Is there a ``simpler'' criterion for stability?
	
	``Simpler'' is up to interpretation. It can mean a criterion computationally verifiable in polynomial time, or a criterion resembling that of Mairesse and Moyal \cite[Theorem 1]{Mairesse2016} or that of Comte et al. \cite[Proposition 7 and 8]{Comte2021}.
\end{itemize}

In this Section, we provide answers to both questions above. We show that \eqref{eq:onlcond} and \eqref{eq:gencond} are in fact equivalent. Moreover, we devise another stability criterion equivalent to both \eqref{eq:onlcond} and \eqref{eq:gencond}, and which is a generalisation of that of Comte et al.

First, we recall the stability criterion by Comte et al. \cite[Proposition 7 and 8]{Comte2021}.

\begin{proposition}
	If $G$ is a graph, then $(G, \mu)$ is stabilisable if and only if the equation $A \lambda = \mu$ has a positive solution $\lambda \in \mathbb{R}_{> 0}^E$ and $A$ is surjective.
\end{proposition}

Similarly, we show that the following condition to be necessary and sufficient.
\begin{equation}
	\label{eq:inccond}
	\textrm{\parbox{.8\textwidth}{Let $A$ be the incidence matrix of the hypergraph $G$, then $A \lambda = \mu$ has a positive solution $\lambda \in \mathbb{R}_{> 0}^E$ and $A$ is surjective.
	}}
\end{equation}

First, we show that \eqref{eq:gencond} implies \eqref{eq:inccond}, which implies necessity.

\begin{lemma}
	\label{lemma: gencond implies inccond}
	Condition \eqref{eq:gencond} implies Condition \eqref{eq:inccond}.
\end{lemma}
\begin{proof}
	First, to show the surjectivity of $A$, we note that for a given hypergraph $G$, the set of arrival rates $\mu$ satisfying \eqref{eq:gencond} is open.
	
	Indeed, fix $\mu$ such that $(G, \mu)$ is stabilisable. For each support $U$, let $\nu^U$, $\gamma^U$, and $0 \leq \alpha^U \leq 1$ be an assignment rate, a reassignment rate, and a parameter satisfying \eqref{eq:gencond}. Let
	\[
		\varepsilon_U = \min_{(v, e) \in U} \frac{1}{|e|} \left[\alpha^U \nu^U_e + \left(1 - \alpha^U\right) \Gamma_e\right] - \left[\alpha^U \nu^U_{v, e} + \left(1 - \alpha^U\right) \Gamma_{u, e}\right] > 0
	\]
	and let $\varepsilon = \min_U \varepsilon_U$.
	
	Moreover, for each support $U$ and each vertex $u$, choose an arbitrary hyperedge $e^U_u$ containing $u$ with $\nu^U_{u, e} > 0$, which must exist as $\sum_{e \ni u} \nu_{u, e} = \mu_u > 0$. Denote $\varepsilon^U_\nu = \min_{u} \nu_{u, e^U_u}$, and $\varepsilon_\nu = \min_U \varepsilon^U_\nu$.
	
	Now consider a given $\eta \in \mathbb{R}_{\geq 0}^V$ such that $\|\eta - \mu\|_1 = \max_{u \in V} |\eta_v - \mu_v| \leq \min \left(\varepsilon_\nu, \frac{\varepsilon_G}{3}\right)$. For a given support $U$, let $\nu = \nu^U$. For each $(u, e) \in G$, let
	\[
		\nu'_{u, e} = 
		\begin{cases}
			\nu_{u, e} + \eta_ u- \mu_u & \text{ if } e = e_u \\
			\nu_{u, e} & \text{ otherwise}
		\end{cases},
	\]
	we have that $\nu' \in \mathbb{R}_{\geq 0}^G$ is an assignment rate, as for each $u \in V$, we have $\sum_{e \ni u} \nu'_{u, e} = \eta_u - \mu_u + \sum_{e \ni u} \nu_{u, e} = \eta_u$. 
	
	Finally, for all $(u, e) \in G$, we have $|\nu_{u, e} - \nu'_{u, e}| \leq \|\eta - \mu\|_1 \leq \frac{\varepsilon_G}{3}$, implying
	\[
		|\nu'_e - \nu_e| = \left| \sum_{u \in e} \nu'_{u, e} - \nu_{u, e}\right| \leq \sum_{u \in e} |\nu'_{u, e} - \nu_{u, e}| \leq |e| \frac{\varepsilon_G}{3}.
	\]
	In particular, this implies $\frac{\nu_e}{|e|} \leq \frac{\nu'_e}{|e|} + \frac{\varepsilon_G}{3}$, thus for each $(u, e) \in U$, we have
	\begin{align*}
		\alpha \nu'_{v, e} + (1 - \alpha) \Gamma_{v, e}
		& \leq \alpha\frac{\varepsilon_G}{3} + \alpha \nu_{v, e} + (1 - \alpha) \Gamma_{u, e} \\
		& \leq \alpha\frac{\varepsilon_G}{3} - \varepsilon_G + \frac{1}{|e|} \left[ \alpha \nu_e + (1 - \alpha) \Gamma_e\right] \\
		& \leq 2\alpha\frac{\varepsilon_G}{3} - \varepsilon_G + \frac{1}{|e|} \left[ \alpha \nu'_e + (1 - \alpha) \Gamma_e\right] \\
		& \leq -\alpha\frac{\varepsilon_G}{3} + \frac{1}{|e|} \left[ \alpha \nu'_e + (1 - \alpha) \Gamma_e\right] \\
	\end{align*}
	Therefore, $\nu'$, $\gamma$, and $\alpha$ satisfy \eqref{eq:gencond} for $U$ and $\eta$. Since this argument holds for all $U$, $(G, \eta)$ is stabilisable.
	
	Going back to conservation equation, since $(G, \mu)$ and $(G, \eta)$ are both stabilisable, there exist vectors $\lambda_\mu$ and $\lambda_\eta$ such that $A \lambda_\mu = \mu$ and $A \lambda_\eta = \eta$. This implies 
	\[
		A (\lambda_\eta - \lambda_\mu) = \eta - \mu
	\]
	but since this is true for all $\|\eta - \mu\|_1 \leq \min \left(\varepsilon_\nu, \frac{\varepsilon_G}{3}\right)$, we conclude that $A$ is surjective.
	
	Next, $A \lambda = \mu$ having a \textit{non-negative} solution $\lambda_0 \in \mathbb{R}^E_{\geq 0}$ is simply the conservation equation which holds when $(G, \mu)$ is stabilisable. However, to construct a \textit{positive} solution $\lambda \in \mathbb{R}^E_{> 0}$, we need an extra argument.
	
	Let $0 < \varepsilon < \frac{1}{|G|} \min \left(\varepsilon_\nu, \frac{\varepsilon_G}{3}\right)$, and $\mu_\varepsilon \in \mathbb{R}^V$ be defined by $(\mu_\varepsilon)_u = \mu_u - \sum_{e \ni u} \varepsilon$. We have that $\| \mu_\varepsilon - \mu\|_1 = \varepsilon |G| < \min \left(\varepsilon_\nu, \frac{\varepsilon_G}{3}\right)$ by construction, so by the argument above, we have $(G, \mu_\varepsilon)$ is stabilisable. Thus, it satisfies the conservation equation, meaning there exists $\lambda_\varepsilon \in \mathbb{R}^E_{\geq 0}$ such that $A \lambda_\varepsilon = \mu_\varepsilon$.
	
	Let $\lambda = \lambda_\varepsilon + \varepsilon \mathbbm{1}$, then by construction, we have $A\lambda = \mu$, and moreover, for all $e \in E$, we also have $\lambda_e = \lambda_\varepsilon + \varepsilon \geq \varepsilon > 0$, as desired.
\end{proof}

Then, we show that \eqref{eq:inccond} implies \eqref{eq:onlcond}, which implies sufficiency.

\begin{lemma}
	Condition \eqref{eq:inccond} implies Condition \eqref{eq:onlcond}.
\end{lemma}
\begin{proof}
	Our starting point is the solution $\lambda \in \mathbb{R}_{> 0}^E$ to the equation $A \lambda = \mu$.
	
	Let $U$ be a support. For some $\varepsilon > 0$ to be defined later, let 
	\[
	\zeta_{u, e} = \begin{cases}
		\varepsilon & \text{ if } (u, e) \in U \\
		0 & \text{ otherwise}
	\end{cases}.
	\]
	This induces a vector $\eta \in \mathbb{R}^V$ by $\eta_u = \sum_{e \ni u} \zeta_{u, e}$. Since $A$ is surjective, there exists $\xi \in \mathbb{R}^E$ such that $A \xi = \eta$, or in other words, for all $u \in V$, $\eta_u = \sum_{e \ni u} \xi_e$.
	
	Now for $(u, e) \in G$, let $\nu_{u, e} = \lambda_e - \zeta_{u, e} + \xi_e$, we will show that for suitably chosen $\varepsilon$, $\nu$ is an assignment rate satisfying \eqref{eq:onlcond} for $U$.
	\begin{itemize}
		\item By construction, for all $u \in V$, we have
		\[
		\begin{aligned}
			\sum_{e \ni u} \nu_{u, e}
			& = \sum_{e \ni u} \lambda_e - \zeta_{u, e} + \xi_e = \sum_{e \ni u} \lambda_e -  \sum_{e \ni u}\zeta_{u, e} +  \sum_{e \ni u}\xi_e\\
			& = \mu_u - \eta_u + \eta_u = \mu_u;
		\end{aligned} 
		\]
		\item Moreover, we have $-\zeta_{u, e} < -\frac{1}{|e|} \zeta_e$ where $\zeta_e = \sum_{v \in e} \zeta_{v, e}$. This is because of the definition of support, where for all $e$, there exists $v' \in e$ such that $(v', e) \not\in U$, so $\zeta_e \leq (|e|-1) \varepsilon$.
		
		In particular, this implies
		\begin{align*}
			\nu_{u, e} 
			& = \lambda_e - \zeta_{u, e} + \xi_e = -\zeta_{u, e} + \frac{1}{|e|} \sum_{v \in e} \lambda_e + \xi_e \\
			& \leq -\frac{1}{|e|} \sum_{v \in e} \zeta_{v, e} + \frac{1}{|e|} \sum_{v \in e} \lambda_e + \xi_e = \frac{1}{|e|} \sum_{u \in e} (\lambda_e - \zeta_{v, e} + \xi_e) \\
			& = \frac{1}{|e|} \sum_{v \in e} \nu_{v, e} = \frac{\nu_e}{|e|};
		\end{align*}
		\item Finally, note that $\|\eta\|_1 = |U|\varepsilon$. Since $A$ is a linear map, it is continuous, hence, for $\varepsilon > 0$ small enough, we have that $\|\xi\|_\infty \leq -\varepsilon + \min_{(u, e) \in G} \lambda_e$. This implies that for all $(u, e) \in G$.
		\[
			\nu_{u, e} = \lambda_e - \zeta_{u, e} + \xi_e \geq \lambda_e - \varepsilon - \|\xi\|_\infty \geq 0.
		\]
	\end{itemize}
	These three properties means that $\nu$ is an assignment rate satisfying \eqref{eq:onlcond} for $U$. Since this argument holds for all $U$, we have that \eqref{eq:onlcond} is satisfied, as desired.
\end{proof}

Combining the two lemmata with the fact that \eqref{eq:onlcond} is a special case of \eqref{eq:gencond}, we get that 
\begin{theorem}
	\label{theorem: necessity and sufficiency of inccond}
	Condition \eqref{eq:onlcond}, \eqref{eq:gencond}, and \eqref{eq:inccond} are equivalent.
\end{theorem}
Consequently, the stationary size-based randomised online assignment policy constructed in the proof of Lemma \ref{lemma: sufficience, online assignment} in fact admits maximum stability region.
	\section{\texttt{NCOND} versus Condition \eqref{eq:onlcond}}
\label{Graphs vs hypergraphs}

In Section \ref{Criterion for online assignment policies}, we have seen the intuition behind \texttt{NCOND} in graph case, which leads to our conjectured Condition \eqref{eq:onlcond}. Whilst the intuition does not trivially get carried over to hypergraph case, it does lead to a correct guess. Moreover, as we seen in Section \ref{Equivalence}, Condition \eqref{eq:onlcond} is indeed a stability criterion.

Given the similarity between \texttt{NCOND} and Condition \eqref{eq:onlcond}, as graph case is a special case of hypergraph case, one may pose a natural question: is \texttt{NCOND} is also a special case of Condition \eqref{eq:onlcond}? Of course, they are equivalent since they are both necessary and sufficient for stability, but is there a way to directly derive \texttt{NCOND} from \eqref{eq:onlcond}?

In this Section, we show that the answer is yes, and hope that the derivation clears up the relation between the two conditions, as well as the rather informal intuition presented in Section \ref{Criterion for online assignment policies}.

\begin{lemma}
	If $G$ is a graph, then \texttt{NCOND} and Condition \eqref{eq:onlcond} are equivalent.
\end{lemma}
\begin{proof}
	Every independent set $U \subsetneq V$ of $G$ corresponds to a support $W$: for each $u$ in $U$ and each edge $e$ containing $u$, we add $(u, e)$ to $W$. By Condition \eqref{eq:onlcond}, there exists an assignment rate $\nu$ such that for all $(u, e) \in W$, we have
	\[
		\nu_{u, e} < \frac{\nu_e}{|e|} = \frac{\nu_e}{2}.
	\]
	If $e = (u, v)$ is the edge between $u$ and $v$, then since $\nu_e = \nu_{u, e} + \nu_{v, e}$, this can be rewritten as $\nu_{u, e} < \nu_{v, e}$. Summing over all such edges, we get
	\begin{align*}
		\mu(U)
		& = \sum_{u \in U} \mu_u = \sum_{u \in U} \sum_{e = (u, v)} \nu_{u, e} < \sum_{u \in U} \sum_{e = (u, v)} \nu_{v, e} = \sum_{v \in N(U)} \sum_{e = (u, v), u \in U} \nu_{v, e} \\
		& < \sum_{v \in N(U)} \sum_{e = (u, v)} \nu_{v, e} = \sum_{v \in N(U)} \mu_v = \mu(N(U)),
	\end{align*}
	meaning that Condition \eqref{eq:onlcond} implies \texttt{NCOND}.
	
	The converse is more difficult, mainly because not all support sets are independent sets: it can happen that for a support $W$, we have $(u, e)$ in $W$, but also $(v, f)$ in $W$ for $f = (u, v)$. This corresponds to having class-$u$ items assigned to $e$ and class-$v$ items assigned to $f$, but no matching of type-$f$ is possible, since Condition \eqref{eq:onlcond} does not allow reassignment.
	
	That being said, given a support $W$, we classify the vertices into three categories:
	\begin{itemize}
		\item Category 1: $u$ such that $(u, e) \in W$ for some edge $e$, and there is no neighbour $v$ of $u$ such that $(v, f) \in W$ for $f = (u, v)$;
		\item Category 2: $u$ such that $(u, e) \in W$ for some edge $e$, and there are some neighbours $v$ of $u$ such that $(v, f) \in W$ for $f = (u, v)$;
		\item Category 3: $u$ such that $(u, e) \not\in W$ for all edges $e$ containing $u$.
	\end{itemize}
	
	We now need to construct an assignment rate $\nu$ guaranteeing the inequality condition in \eqref{eq:onlcond}.
	
	First, by definition, the vertices in Category 1 form an independent set, which we will denote by $U$. \texttt{NCOND} implies that $\mu(N(U)) > \mu(U)$.
	
	Let us consider the bipartite graph $H = (A \cup B, F)$ induced by $U$ and $N(U)$: namely, let $A = U$, $B = N(U)$, and $F$ be the set of edges between $U$ and $N(U)$ with infinite capacity. To this, we add a source $s$ and a sink $t$. To each vertex $u$ in $A$, we add an edge from $s$ of capacity $\mu_u$; from each vertex $v$ in $B$, we add an edge to $t$ of capacity $\mu_v$.
	
	Now let us see what a min cut will be in this flow network: naturally we cannot cut any edge between $A$ and $B$. Suppose we cut edges between $s$ and $A' \subseteq A$, it remains to cut edges between $t$ and $B' = N(A \setminus A')$. The cost of this cut is
	\begin{align*}
		\mu(A') + \mu(B')
		& = \mu(A') + \mu(N(A \setminus A')) \\
		& \geq \mu(A') + \mu(A \setminus A') = \mu(A') + \mu(A) - \mu(A') = \mu(A)
	\end{align*}
	with the inequality due to applying \texttt{NCOND} to $A \setminus A'$ which is either empty or independent. In particular, the inequality is strict whenever $A \setminus A'$ is non-empty, so the equality occurs if and only if $A' = A$.
	
	By max-flow - min-cut theorem, this is the maximum flow, and moreover, there exists a maximum flow such that all other edges are not saturated. In particular, for all vertices $v \in B$, the incoming flow toward $v$ does not saturate $v$. 
	
	Take a maximum flow, for all edges $e = (u, v)$ between $u \in A$ and $v \in B$, let $\nu_{u, e}$ be the flow of edge $e$. Given a vertex $v \in B$, since $v$ is not saturated, we have $\mu_v - \sum_{e = (u, v)} \nu_{u, e} > 0$. Distribute this slack equally, so that we have $\nu_{v, e} > \nu_{u, e}$ for all $e = (u, v)$. This guarantees the inequality condition in \eqref{eq:onlcond} for all edges involving vertices in Category $1$.
	
	Finally, the remaining edges $e = (u, v)$ concerned by the inequality condition must involve either vertices $u$ and $v$ in Category 2 and/or 3. Moreover, $e \in W$, so exactly one of the two vertices $u$ and $v$ are in Category 2. Without loss of generality, assume that it is $u$, so $v$ must be in Category 3.
	
	Then, we set $\nu_{u, e} = 0$. By doing this, the arrival rate $\mu_u > 0$ will be distributed amongst the neighbours $v$ such that $(v, f) \in W$ for $f = (u, v)$; by definition, these neighbours exist, so this is possible. It remains to guarantee that $\nu_{v, e} > 0$. This can be done, for example by setting $\nu_{v, e} = \lambda_e$ where $\lambda \in \mathbb{R}_{> 0}^E$ is a solution of $A \lambda = \mu$, which is guaranteed to exist by the conservation equation.	
\end{proof}

This proof also demonstrates how even though Condition \eqref{eq:onlcond} in theory can be derived directly from \texttt{NCOND}, the derivation is not trivial and somewhat not natural. This partially explains why various attempts by Rahme and Moyal \cite{Rahme2021} to generalise \texttt{NCOND} in looking for a stability criterion have not been successful.

	\section{Conclusion}
\label{Conclusion}

We have introduced online assignment policies and and showed that they are a right generalisation to greedy polices in the hypergraph setting. In particular, this policy family is maximally stable, and the analysis leads to three equivalent necessary and sufficient criteria for stability, namely Condition \eqref{eq:onlcond}, \eqref{eq:gencond}, and \eqref{eq:inccond}, which generalise known criteria for graph setting. This thus completely answers the question of such criterion posed by Rahme and Moyal \cite{Rahme2021}.

Moreover, the proof of sufficiency is constructive, based on carefully designed Lyapunov function $L_+$, which leads to a $L_+$-MaxWeight policy that is maximally stable, arrival-rate agnostic, and Markovian. On the other hand, for necessity, we develop a new framework which extends the law-of-large-numbers argument used to prove necessity in case of stochastic matching on graphs.

What is more valuable than the stability criteria introduced in this paper is the framework of online assignment polices, as well as the ideas in the proofs, which can be generalised to other cases of stochastic matching.

In particular, one may allow batch arrivals, meaning multiple items arrive at the same time, and weighted matching - that is, a matching might require more than one item of a given time. In fact, one can allow the weights to be negative, which corresponds to routing as commonly found in assemble-to-order systems. Another direction would be to introduce service time - that is, the time a matching takes before leaving the buffer - and abandonment. The former is also commonly found in production lines, and the latter is a natural phenomenon in real-life applications. These generalisations will be addressed in future works.
	
	\bibliographystyle{ieeetr}
	\bibliography{refs}

@phdthesis{Xie2022,
author = {Xie, Bowen},
school = {Iowa State University},
title = {{Topics of Queueing Theory in Heavy Traffic}},
url = {https://www.proquest.com/openview/00d4051b819cafeb6ab9d636487b39d5},
year = {2022}
}

@article{Gurvich2014,
author = {Gurvich, Itai and Ward, Amy},
doi = {10.1287/13-ssy097},
issn = {1946-5238},
journal = {Stochastic Systems},
number = {2},
pages = {479--523},
title = {{On the Dynamic Control of Matching Queues}},
volume = {4},
year = {2014}
}

@article{Teubner2015,
author = {Teubner, Timm and Flath, Christoph M.},
doi = {10.1007/s12599-015-0396-y},
issn = {2363-7005},
journal = {Business \& Information Systems Engineering},
number = {5},
pages = {311--324},
title = {{The Economics of Multi-Hop Ride Sharing}},
volume = {57},
year = {2015}
}

@article{Comte2021,
archivePrefix = {arXiv},
arxivId = {2112.14457},
author = {Comte, C{\'{e}}line and Mathieu, Fabien and Varma, Sushil Mahavir and Bu{\v{s}}i{\'{c}}, Ana},
eprint = {2112.14457},
title = {{Online Stochastic Matching: A Polytope Perspective}},
url = {http://arxiv.org/abs/2112.14457},
year = {2021}
}

@article{Rahme2021,
archivePrefix = {arXiv},
arxivId = {1907.12711},
author = {Rahme, Youssef and Moyal, Pascal},
doi = {10.1017/apr.2021.8},
eprint = {1907.12711},
issn = {00018678},
journal = {Advances in Applied Probability},
number = {4},
pages = {951--980},
title = {{A stochastic matching model on hypergraphs}},
volume = {53},
year = {2021}
}

@article{Veron2014,
author = {V{\'{e}}ron, Maxime and Marin, Olivier and Monnet, S{\'{e}}bastien},
doi = {10.1145/2578260.2578265},
isbn = {9781450327060},
journal = {Proceedings of the 24th ACM Workshop on Network and Operating Systems Support for Digital Audio and Video, NOSSDAV 2014},
pages = {7--12},
title = {{Matchmaking in multi-player on-line games: Studying user traces to improve the user experience}},
year = {2014}
}

@article{Busic2013,
archivePrefix = {arXiv},
arxivId = {1003.3477},
author = {Bu{\v{s}}i{\'{c}}, Ana and Gupta, Varun and Mairesse, Jean},
doi = {10.1239/aap/1370870122},
eprint = {1003.3477},
issn = {00018678},
journal = {Advances in Applied Probability},
number = {2},
pages = {351--378},
title = {{Stability of the bipartite matching model}},
volume = {45},
year = {2013}
}

@article{Roth2007,
author = {Roth, Alvin E. and S{\"{o}}nmez, Tayfun and {Utku {\"{U}}nver}, M.},
doi = {10.1257/aer.97.3.828},
issn = {00028282},
journal = {American Economic Review},
number = {3},
pages = {828--851},
pmid = {29135211},
title = {{Efficient kidney exchange: Coincidence of wants in markets with compatibility-based preferences}},
volume = {97},
year = {2007}
}

@article{Gupta2024,
author = {Gupta, Varun},
journal = {Operations Research},
number = {3},
pages = {1139},
title = {{Algorithm for Multiway Matching with Bounded Regret}},
volume = {72},
year = {2024}
}

@article{Mairesse2016,
author = {Mairesse, Jean and Moyal, Pascal},
doi = {10.1017/jpr.2016.65},
issn = {0021-9002},
journal = {Journal of Applied Probability},
month = {dec},
number = {4},
pages = {1064--1077},
title = {{Stability of the stochastic matching model}},
url = {https://www.cambridge.org/core/product/identifier/S0021900216000656/type/journal_article},
volume = {53},
year = {2016}
}

@article{Kerimov2021,
author = {Kerimov, S{\"{u}}leyman and Ashlagi, Itai and Gurvich, Itai},
doi = {10.2139/ssrn.3824407},
journal = {SSRN Electronic Journal},
title = {{Dynamic Matching: Characterizing and Achieving Constant Regret}},
year = {2021}
}

@article{Nazari2019,
archivePrefix = {arXiv},
arxivId = {1608.01646},
author = {Nazari, Mohammadreza and Stolyar, Alexander L.},
doi = {10.1007/s11134-018-9593-y},
eprint = {1608.01646},
issn = {15729443},
journal = {Queueing Systems},
number = {1-2},
pages = {143--170},
title = {{Reward maximization in general dynamic matching systems}},
volume = {91},
year = {2019}
}

@article{Harrison1973,
author = {Harrison, J. Michael},
doi = {10.1017/s0021900200095358},
issn = {0021-9002},
journal = {Journal of Applied Probability},
number = {02},
pages = {354--367},
title = {{Assembly-like queues}},
volume = {10},
year = {1973}
}

@article{Mathieu2026,
archivePrefix = {arXiv},
arxivId = {2607.18935},
author = {Mathieu, Fabien},
eprint = {2607.18935},
month = {jul},
title = {{Online Stochastic Matchings: Stability on Hypergraphs}},
url = {http://arxiv.org/abs/2607.18935},
year = {2026}
}

@misc{Busic2026,
	author = {Bu{\v{s}}i{\'{c}}, Ana and Doncel, Josu and Fourneau, Jean-Michel and Li, Shu and Nguyen, Doan Dai},
	title = {{Dynamic matching models under greedy policies: performance paradox and extension to hypergraphs}},
	note = {\url{https://github.com/enbugging/littlenotebook/blob/main/assets/IFORS 2026 talk.pdf}}, 
	url = {https://github.com/enbugging/littlenotebook/blob/main/assets/IFORS 2026 talk.pdf},
	year = {2026}
}

@misc{Nguyen2026,
	author = {Nguyen, Doan Dai},
	title = {{Stability of stochastic matchine on hypergraphs}},
	note = {\url{https://www.youtube.com/watch?v=yEMXmD3aYrg}}, 
	url = {https://www.youtube.com/watch?v=yEMXmD3aYrg},
	year = {2026}
}
\end{document}